\documentclass[11pt]{article}

\usepackage{times} 
\usepackage{epsfig}
\usepackage{amsmath}
\usepackage{amssymb}
\usepackage{amsthm}
\usepackage{graphics}
\usepackage{psfrag}

\newtheorem{theorem}{Theorem}
\newtheorem{lemma}{Lemma}
\newtheorem{corollary}{Corollary}
\newtheorem{proposition}{Proposition}
\newtheorem{observation}{Observation}

\newcommand{\old}[1]{{}}
\newcommand{\later}[1]{{}}

\newcommand{\eps}{\varepsilon}
\newcommand{\NN}{\mathbb{N}}

\newcommand{\RR}{\mathbb{R}}

\def\B{\mathcal B}

\def\H{\mathcal H}
\def\I{\mathcal I}
\def\X{\mathcal X}

\newcommand{\vol}{{\rm vol}}

\title{On the largest empty axis-parallel box amidst $n$ points}

\author{%
Adrian Dumitrescu\thanks{Department of Computer Science,
University of Wisconsin--Milwaukee,
WI 53201-0784, USA\@. Email: \texttt{ad@cs.uwm.edu}.
Supported in part by NSF CAREER grant CCF-0444188.
Part of the research by this author was done at the
Ecole Polytechnique F\'ed\'erale de Lausanne.}
\and
Minghui Jiang\thanks{Department of Computer Science, 
Utah State University, Logan, UT 84322-4205, USA\@.
Email: \texttt{mjiang@cc.usu.edu}.
Supported in part by NSF grant DBI-0743670.}}

\begin{document}

\maketitle

\thispagestyle{empty}

\begin{abstract}
We give the first nontrivial upper and lower bounds on the maximum volume
of an empty axis-parallel box inside an axis-parallel unit hypercube
in $\RR^d$ containing $n$ points. 
For a fixed $d$, we show that the maximum volume is of the order
$\Theta\left(\frac{1}{n}\right)$.  We then  use the fact that the maximum volume is
$\Omega\left(\frac{1}{n}\right)$ in our design of the first efficient
$(1-\eps)$-approximation algorithm for the following problem: 
Given an axis-parallel $d$-dimensional box $R$ in $\RR^d$
containing $n$ points, compute a maximum-volume empty axis-parallel
$d$-dimensional box contained in $R$. The running time of our
algorithm is nearly linear in $n$, for small $d$, and increases only
by an $O(\log{n})$ factor when one goes up one dimension. 
No previous efficient exact or approximation algorithms were known for this
problem for $d \geq 4$. As the problem has been recently shown to be
NP-hard in arbitrary high dimensions (i.e., when $d$ is part of the
input), the existence of efficient exact algorithms is unlikely.

We also obtain tight estimates on the maximum volume of an empty
axis-parallel hypercube inside an axis-parallel unit hypercube
in $\RR^d$ containing $n$ points. For a fixed $d$, 
this maximum volume is of the same order order $\Theta\left(\frac{1}{n}\right)$.  
A faster $(1-\eps)$-approximation algorithm, with a  milder dependence
on $d$ in the running time, is obtained in this case.  
\end{abstract}

\medskip
\hspace{0.15in}

\newpage
\setcounter{page}{1}
\setcounter{footnote}{0}

\section{Introduction}

Given a set $S$ of $n$ points in the unit square $U=[0,1]^2$,
let $A(S)$ be the maximum area of an empty axis-parallel rectangle
contained in $U$, and let $A(n)$ be the minimum value
of $A(S)$ over all sets $S$ of $n$ points in $U$.
For any dimension $d \geq 2$, given a set $S$ of $n$ points in the unit 
hypercube $U_d=[0,1]^d$, let $A_d(S)$ as the maximum volume of an
empty axis-parallel hyperrectangle ($d$-dimensional axis-parallel box)
contained in $U_d$, and let $A_d(n)$ be the minimum value of $A_d(S)$
over all sets $S$ of $n$ points in $U_d$.
For simplicity we sometimes omit the subscript $d$ in the planar case ($d=2$). 

In this paper we give the first nontrivial upper and lower bounds on $A_d(n)$.
For any dimension $d$, our estimates are within a multiplicative
constant (depending on $d$) from each other. For a fixed $d$, we show
that the maximum volume is of the order $\Theta(\frac{1}{n})$.   
While the algorithmic problem of finding an empty axis-parallel box
of maximum volume has been previously studied for $d=2,3$ (see below),
estimating the maximum volume of such a box as a function of $d$ and
$n$ seems to have not been previously considered. 

We first introduce some notations and definitions. Throughout this paper,
a \emph{box} is an \emph{open} axis-parallel hyperrectangle
contained in the unit hypercube $U_d=[0,1]^d$, $d \ge 2$.
Given a set $S$ of points in $U_d$,
a box $B$ is \emph{empty} if it contains no points in $S$,
i.e., $B\cap S = \emptyset$.
If $B$ is a box, we also refer to the side length of $B$ in the $i$th
coordinate as the extent in the $i$th coordinate of $B$. 
Throughout this paper, $\log{n}$ and $\ln{n}$ are the logarithms of $n$ in
base $2$ and $e$, respectively.

Given an axis-parallel rectangle $R$ in the plane containing $n$ points,
the problem of computing a maximum-area empty axis-parallel
sub-rectangle contained in $R$ is one of the oldest problems studied 
in computational geometry. For instance, this problem arises 
when a rectangular shaped facility is to be located within a
similar region which has a number of forbidden areas, or in cutting
out a rectangular piece from a large similarly shaped metal sheet with
some defective spots to be avoided~\cite{NLH84}. In higher dimensions,
finding the largest empty axis-parallel box has applications
in data mining, in finding large gaps in a multi-dimensional data
set~\cite{EGLM03}. 

Several algorithms have been proposed for the planar problem
over the years \cite{AS87,AF86,AK89,CDL86,D92,MOS85,NLH84,O90}. 
For instance, an early algorithm by Chazelle, Drysdale and Lee
\cite{CDL86} runs in $O(n \log^3{n})$ time and $O(n \log{n})$ space.
The fastest known algorithm,
proposed by Aggarwal and Suri in 1987 \cite{AS87},
runs in $O(n \log^2{n})$ time and $O(n)$ space. 
A lower bound of $\Omega(n \log{n})$ in the algebraic decision tree
model for this problem has been shown by Mckenna et~al.\ \cite{MOS85}.

For any dimension $d$, there is an obvious brute-force algorithm
running in $O(n^{2d+1})$ time and $O(n)$ space. 
No significantly faster algorithms, i.e., with a fixed degree
polynomial running time in $\RR^d$, where known.
Confirming this state of affairs,
Backer and Keil~\cite{BK09a,BK09b} recently proved that the problem is
NP-hard in arbitrary high dimensions (i.e., when $d$ is part of the input).
They also gave an exact algorithm running in $O(n^d \log^{d-2}{n})$
time, for any $d \geq 3$.  
In particular, the running time of their exact algorithm for $d=3$ is
$O(n^3 \log{n})$. Previously, Datta and Soundaralakshmi~\cite{DS00}
had reported an $O(n^3)$ time exact algorithm for the $d=3$ case, but
their analysis for the running time seems incomplete. Specifically,
the $O(n^3)$ running time depends on an $O(n^3)$ upper bound on the
number of maximal empty boxes (see discussions in the next paragraph),
but they only gave an $\Omega(n^3)$ lower bound.
Here we present the first efficient $(1-\eps)$-approximation algorithm 
for finding an axis-parallel empty box of maximum volume, 
whose running time is nearly linear for small $d$, and increases only
by an $O(\log{n})$ factor when one goes up one dimension. 

An empty box of maximum volume must be maximal with respect to inclusion. 
Following the terminology in \cite{NLH84},
a {\em maximal empty} box is called {\em restricted}. 
Thus the maximum-volume empty box in $U_d$ is restricted.
Naamad et~al.\ \cite{NLH84} have shown that in the plane, the number of 
restricted rectangles is $O(n^2)$, and that this bound is tight.
It was conjectured by Datta and Soundaralakshmi~\cite{DS00} that the
maximum number of restricted boxes is $O(n^{d})$ for each (fixed) $d$.  
The conjecture has been recently confirmed by Backer and
Keil~\cite{BK09a,BK09b} (for $d \geq 3$). 
Here we extend (Theorem~\ref{thm:restricted}, Appendix~\ref{sec:restricted})
the constructions with $\Omega(n^d)$ restricted boxes for $d=2$
in \cite{NLH84} and $d=3$ in \cite{DS00} for arbitrary $d$. 
Independently and simultaneously, Backer and Keil have also obtained
this result~\cite{Ba09,BK09a,BK09b}. Hence the maximum number of restricted
boxes is $\Theta(n^{d})$ for each fixed $d$.   
This means that any algorithm for computing a maximum-volume empty box 
based on enumerating restricted boxes is inefficient in the worst case.
On the other hand, at the expense of giving an $(1-\eps)$-approximation, 
our algorithm does not enumerate all restricted boxes, and
achieves efficiency by enumerating all canonical boxes (to be defined)
instead.

%\vspace{-0.5\baselineskip}
\paragraph{Our results are:}\quad
\begin{description}
\item[(I)] In Section~\ref{sec:rectangle} we show 
that $A_d(n)=\Theta\left(\frac{1}{n}\right)$ for $d \geq 2$. More precisely: 
$A_d(n) \geq \frac{1}{n+1}$, and
$A_d(n) \geq \left(\frac{5}{4}-o(1)\right) \cdot \frac{1}{n}$. 
From the other direction we have $A_2(n) < 4\cdot \frac{1}{n}$, and
$A_d(n) < (2^{d-1} \prod_{i=1}^{d-1} p_i) \cdot \frac{1}{n}$ for any
$d \geq 3$. Here $p_i$ is the $i$th prime. 

\item[(II)] In Section~\ref{sec:approx1} we exploit the fact that
the maximum volume is $\Omega\left(\frac{1}{n}\right)$ in our design
of the first efficient 
$(1-\eps)$-approximation algorithm for finding the largest empty box:
Given an axis-parallel $d$-dimensional box $R$ in $\RR^d$ containing
$n$ points, there is a $(1-\eps)$-approximation algorithm, running in  \linebreak
$O((8 e d \eps^{-2})^d \cdot n \log^d{n})$ time, 
for computing a maximum-volume empty axis-parallel box
contained in $R$. 

\item[(III)] In Appendix~\ref{sec:square} we show that the
$\Theta\left(\frac{1}{n}\right)$  estimate also holds for the maximum volume
(or area) of an axis-aligned hypercube (or square) amidst $n$ point in $[0,1]^d$. 
In Appendix~\ref{sec:approx2} we present a faster
$(1-\eps)$-approximation algorithm for finding the largest empty
hypercube: Given an axis-parallel $d$-dimensional hypercube $R$ in $\RR^d$ containing
$n$ points, there is a $(1-\eps)$-approximation algorithm, running in  
$ O( d^2 \eps^{-1} \cdot n \log{n} +
( 4d \eps^{-1})^{d+1} \cdot n^{1/d} \log{n}) $ time, 
for computing a maximum-volume empty axis-parallel hypercube
contained in $R$. 

\item[(IV)] In Appendix~\ref{sec:restricted} we derive
an $\Omega(n^d)$ lower bound on the number of restricted boxes in
$d$-space, for fixed $d$. This matches the recent $O(n^d)$ upper bound
of Backer and Keil~\cite{BK09a,BK09b}. Following their idea, we further
narrow the gap between the bounds (in the dependence of $d$) 
based on a finer estimation. 

\end{description}

\section{Empty rectangles and boxes}\label{sec:rectangle}

\subsection{Empty rectangles in the plane}\label{subsec:plane}

%\vspace{-0.5\baselineskip}
\paragraph{The lower bound.}
We start with a very simple-minded lower bound; however, as it turns
out, it is very close to optimal.
One can immediately see that $A(n) =\Omega(\frac{1}{n})$, by
partitioning the unit square with vertical lines through each point:
out of at most $n+1$ resulting empty rectangles, the largest rectangle
has area at least $\frac{1}{n+1}$. Thus we have:

\begin{proposition} \label{P1}
\begin{equation} \label{E3}
A(n) \geq \frac{1}{n+1}. 
\end{equation}
\end{proposition}

The following observation is immediate from invariance under scaling 
with respect to any of the coordinate axes. 

\begin{observation} \label{O1}
Assume that $A(n) \geq z$ holds for some $n$ and $z$.
Then, given $n$ points in an axis-aligned rectangle $R$, there is an empty
rectangle contained in $R$ of area at least $z \cdot {\rm area}(R)$. 
\end{observation}

Using the next two lemmas we will slightly improve the trivial lower bound
$A(n) \geq \frac{1}{n+1}$ in  our next Theorem~\ref{T2}. 
Let  $\xi=\frac{3-\sqrt{5}}{2}$ be the solution in $(0,1)$ of the
quadratic equation $ (1-x)^2=x$.

\begin{lemma} \label{L1}
Given $2$ points in the unit square, there exists an empty rectangle
with area at least $\frac{3-\sqrt{5}}{2}$. 
This bound is tight, i.e., $A(2)=\frac{3-\sqrt{5}}{2}=0.3819\ldots$.
\end{lemma}
\begin{proof}
Let $p_1,p_2 \in U$, and assume without loss of generality
that $x(p_1) \leq x(p_2)$, and  $y(p_1) \geq y(p_2)$. Write
$x=x(p_1)$, and $y=y(p_2)$. 
Consider the three empty rectangles %\linebreak
%$\{(0,0),(x,0),(x,1),(0,1)\}$, 
%$\{(0,0),(1,0),(1,y),(0,y)\}$, and $\{(x,y),(1,y),(1,1),(x,1)\}$. 
$(0,x)\times(0,1)$,
$(0,1)\times(0,y)$, and
$(x,1)\times(y,1)$.
Their areas are $x$, $y$, and $(1-x)(1-y)$, respectively. 
If $x \geq \xi$ or $y \geq \xi$,
we are done, as one of the first two rectangles has area at least
$\xi$. So we can assume that $x \leq \xi$ and $y \leq \xi$.
Then it follows that
$$ (1-x)(1-y) \geq (1-\xi)^2=\xi, $$
so the third rectangle has area at least $\xi$, as required.

To see that this bound is tight, take $p_1=(\xi,1-\xi)$,
$p_2=(1-\xi,\xi)$, and check that the largest empty rectangle has area $\xi$. 
\end{proof}

The proof of the next lemma appears in Appendix \ref{app:L2}.

\begin{lemma} \label{L2}
Given $4$ points in the unit square, there exists an empty rectangle
with area at least $\frac{1}{4}$. This bound is tight, i.e.,
$A(4)=\frac{1}{4}$.  
\end{lemma}
%
%\begin{proof}

%\end{proof}

\begin{theorem} \label{T2}
Given $n$ points in the unit square, there exists an empty rectangle
with area at least $(\frac{5}{4}-o(1)) \cdot \frac{1}{n}$. 
That is, $A(n) \geq (\frac{5}{4}-o(1)) \cdot \frac{1}{n}$. 
\end{theorem}
\begin{proof}
Write $n=5k+r$, for some $k \in \NN$ and $r \in \{0,1,2,3,4\}$. 
Partition $U$ into $k+1$ rectangles of equal width. There exists
at least one rectangle $R'$ with at most $4$ points in its interior. 
By Lemma~\ref{L2} and Observation~\ref{O1}, 
$R'$ contains an empty rectangle of area at least
$$ \frac{1}{4} \cdot \frac{1}{k+1} \geq \frac{5}{4} \cdot \frac{1}{n+5} 
= \left(\frac{5}{4}-o(1)\right) \cdot \frac{1}{n}, $$
as claimed.
\end{proof}

The lower bound derived in the proof, $\frac{5}{4} \cdot \frac{1}{n+5}$, is
better than $\frac{1}{n+1}$ for all $n \geq 16$. For $n=5k+4$,
the resulting bound is $\frac{5}{4} \cdot \frac{1}{n+1}$. 
An alternative partition, yielding the same bound in Theorem~\ref{T2},
can be obtained by dividing $U$ into rectangles with vertical lines
through every $5$th point of the set. Slightly better lower
bounds, particularly 
for small values of $n$ can be obtained by constructing different
partitions tailored for specific values of $k,r$ (with a number of
points other than $4$ in a few of the rectangles), and using estimates
on $A(2)$, $A(6)$, etc.  For instance, from Lemma~\ref{L2} we can derive that 
$A(6) \geq 3 - 2\sqrt{2}=0.1715\ldots$. 
Incidentally, we remark that a suitable $6$-point construction gives
from the other direction that $A(6)<0.2$.

%\vspace{-0.5\baselineskip}
\paragraph{The upper bound.}
Let $C_n$ be the van der Corput set of $n$ points~\cite{C35a,C35b},
with coordinates $(x(k), y(k))$, $0 \le k \le n-1$, constructed as
follows~\cite{Ch00,Ma99}: Let $x(k) = k/n$.
If $k=\sum_{j \ge 0} a_j 2^j$ is the binary representation of $k$,
where $a_j \in \{0,1\}$, then $y(k) = \sum_{j\ge 0} a_j 2^{-j-1}$.
Observe that all points in $C_n$ lie in the unit square $U=[0,1]^2$.

\begin{theorem} \label{T1}
For the van der Corput set of $n$ points, $C_n \subset U$, the area of
the largest empty axis-parallel rectangle is less than $4/n$.
\end{theorem}
\begin{proof}
Let $B$ be any open empty axis-parallel rectangle inside the unit square.
We next show\footnote{The argument we use here is similar to
that used for bounding the geometric discrepancy of the van der Corput
set of points.} that the area of $B$ is less than $4/n$. 
Following the presentation in \cite[p.~39]{Ma99},
a \emph{canonical} interval is an interval of the form
$[u \cdot 2^{-v}, (u+1) \cdot 2^{-v})$
for some positive integer $v$ and an integer $u \in [0, 2^v -1]$.

Let $I_y = [t \cdot 2^{-q}, (t+1) \cdot 2^{-q})$
be the longest canonical interval
contained in the projection of the empty rectangle $B$ onto the $y$-axis
(recall that $B$ is open, so this projection is an open interval).
Then the side length of $B$ along $y$ must be less than
$2 \cdot 2^{-q+1}$ because otherwise the projection would
contain a longer canonical interval of length $2^{-q+1}$.

Let $k=\sum_{j \ge 0} a_j 2^j$ be the binary representation of
an integer $k$, $0 \le k \le n - 1$.
In the van der Corput construction,
a point in $C_n$ with $x$-coordinate $k/n$
has its $y$-coordinate in the canonical interval $I_y$ if and only if
$t \cdot 2^{-q} \le \sum_{j\ge 0} a_j 2^{-j-1} < (t+1) \cdot 2^{-q}$,
which happens exactly when
$\sum_{j=0}^{q-1} a_j 2^{-j-1} = t \cdot 2^{-q}$.
In this case, $k \bmod 2^q = \sum_{j=0}^{q-1} a_j 2^j$ is a constant
$z= z(t, q)$.
It then follows that the side length of $B$ along $x$
is at most $2^q / n$.
Therefore the area of $B$ is less than
$2 \cdot 2^{-q+1} \cdot 2^q / n =4/n$, as required.
\end{proof}

\begin{corollary} \label{C1}
$A(n) < 4\cdot \frac{1}{n}$.
\end{corollary}

\subsection{Empty boxes in higher dimensions}\label{subsec:higher}

As in the planar case, $A_d(n) \geq \frac{1}{n+1}$ is immediate,
by partitioning the hypercube $U_d$ with $n$ axis-parallel hyperplanes,
one through each of the $n$ points. 
By projecting the $n$ points on one of the faces of $U_d$,
and proceeding by induction on $d$, it follows that the lower bound in
Theorem \ref{T1} carries over here too. 
Thus we have:

\begin{proposition} \label{P2}
$A_d(n) \geq \frac{1}{n+1}$. Moreover,
$A_d(n) \geq \left(\frac{5}{4}-o(1)\right) \cdot \frac{1}{n}$. 
\end{proposition}

\old{
By projecting the $n$ points on one of the faces of $U_d$,
and proceeding by induction on $d$, it follows that the lower bound in
Theorem \ref{T1} carries over here too. 
That is, $A_d(n) \geq (\frac{5}{4}-o(1)) \cdot \frac{1}{n}$. 
} % old

We next show that, modulo a constant factor depending on % the dimension
$d$, this estimate is also best possible. 
Let $H_n$ be the Halton-Hammersely set of $n$
points~\cite{Hal60,Ham60}, with coordinates 
$(x_0(k), x_1(k), \ldots, x_{d-1}(k))$, $0 \le k \le n-1$,
constructed as follows \cite{Ch00,Ma99}:
Let $p_i$ be the $i$th prime number.
Each integer $k$ has a unique base-$p_i$ representation
$k = \sum_{j\ge 0} a_{i,j} p_i^j$, where $a_{i,j} \in [0, p_i-1]$.
Let $x_0(k) = k/n$,
and let $x_i(k) = \sum_{j\ge 0} a_{i,j} p_i^{-j-1}$, $1 \le i \le d-1$.
Then all points in $H_n$ are inside the unit hypercube $U_d=[0,1]^d$.

\begin{theorem} \label{T3}
For the  Halton-Hammersely set of $n$ points, $H_n \subset U_d$, 
the volume of the largest empty axis-parallel box is less than 
$(2^{d-1} \prod_{i=1}^{d-1} p_i) / n$, where $p_i$ is the $i$th prime.
\end{theorem}
\begin{proof}
Let $B$ be any open empty box inside the unit hypercube.
We next show that the volume of $B$
is less than $(2^{d-1} \prod_{i=1}^{d-1} p_i) / n$.
Generalizing the planar case,
a \emph{canonical} interval of the axis $x_i$, $1 \le i \le d - 1$,
is an interval of the form
$[u \cdot p_i^{-v}, (u+1) \cdot p_i^{-v})$
for some positive integer $v$ and an integer $u \in [0, p_i^v - 1]$.
Note that $p_1 = 2$.

First consider each axis $x_i$, $1 \le i \le d-1$.
Let $I_i = [t_i \cdot p_i^{-q_i}, (t_i + 1) \cdot p_i^{-q_i})$
be a longest canonical interval (there could be more than one for
$i \geq 2$)
contained in the projection of the empty box $B$ onto the axis $x_i$.
Then the side length of $B$ along $x_i$ must be less than
$2\cdot p_i^{-q_i+1}$ because otherwise the projection would
contain a longer canonical interval of length $p_i^{-q_i + 1}$.

Next consider the axis $x_0$.
Let $k=\sum_{j \ge 0} a_{i,j} p_i^j$ be the base-$p_i$ representation of
an integer $k$, $0 \le k \le n - 1$ and $1 \le i \le d-1$.
In the Halton-Hammersely construction,
a point in $H_n$ with $x_0$-coordinate $k/n$
has its $x_i$-coordinate in the canonical interval $I_i$ if and only if
$t_i \cdot p_i^{-q_i} \le \sum_{j\ge 0} a_{i,j} p_i^{-j-1}
	< (t_i + 1) \cdot p_i^{-q_i}$,
which happens exactly when
$\sum_{j=0}^{q_i-1} a_{i,j} p_i^{-j-1} = t_i \cdot p_i^{-q_i}$.
In this case, $k \bmod p_i^{q_i} = \sum_{j=0}^{q_i-1} a_{i,j} p_i^j$
is a constant $z_i = z_i(t_i, q_i)$.

Note that the $d-1$ integers $p_i^{q_i}$, $1 \le i \le d-1$,
are relatively prime.
By the Chinese remainder theorem,
it follows that a point in $H_n$ with $x_0$-coordinate $k/n$
has its $x_i$-coordinate in the canonical interval $I_i$ for all
$1 \le i \le d-1$
if and only if
$k \bmod \prod_{i=1}^{d-1} p_i^{q_i} = z$ for some integer
$z = z(t_1, q_1;\ldots;t_{d-1},q_{d-1})$.
Therefore the side length of $B$ along $x_0$
is at most $(\prod_{i=1}^{d-1} p_i^{q_i}) / n$.
Consequently, the volume of $B$ is less than
$(\prod_{i=1}^{d-1} 2 \cdot p_i^{-q_i+1}) \cdot
(\prod_{i=1}^{d-1} p_i^{q_i})/n =
(2^{d-1} \prod_{i=1}^{d-1} p_i) / n$.
\end{proof}

\begin{corollary} \label{C2}
$A_d(n) < (2^{d-1} \prod_{i=1}^{d-1} p_i) \cdot \frac{1}{n}$.
\end{corollary} 

It is known  that $ (\prod_{i=1}^{x} p_i) / x^x \to 1$ as $x \to
\infty$ \cite{R97}.

\section{A $(1-\eps)$-approximation algorithm for finding the largest
empty box}\label{sec:approx1} 

Let $R$ be an axis-parallel $d$-dimensional box in $\RR^d$ 
containing $n$ points. In this section, we present an efficient
$(1-\eps)$-approximation algorithm for computing a maximum-volume
empty axis-parallel box contained in~$R$. 

\begin{theorem} \label{T4}
Given an axis-parallel $d$-dimensional box $R$ in $\RR^d$ containing
$n$ points, there is a $(1-\eps)$-approximation algorithm, running in  
$$ O\left(\left( \frac{8 e d}{\eps^2} \right)^d \cdot n \cdot
\log^d{n} \right) $$
time, for computing a maximum-volume empty axis-parallel box
contained in $R$. 
\end{theorem}

We first set a few parameters. 
%The proof of the theorem is organized as a sequence of lemmas. 

%\vspace{-0.8\baselineskip}
\paragraph {Parameters.}

We assume that $0 <\eps <1$, and $d \geq 3$, 
which cover all cases of interest. To somewhat simplify our
calculations we also assume that $n \geq 12$. Let us choose parameters
\begin{equation} \label{E14}
\delta=\frac{\eps}{2d}, \ \ 
m= \left \lceil \frac{1}{\delta} \right \rceil
=\left \lceil \frac{2d}{\eps} \right \rceil, \ \ {\rm and } \ 
a=\frac{1}{1-\delta}.
\end{equation} 
Let $k$ be the unique positive integer such that
\begin{equation} \label{E28}
a^{k-1} \leq n+1 <a^k.
\end{equation} 

We next derive some inequalities that follow from this setting.
By assumptions $0 < \epsilon < 1$ and $d \ge 3$,
we have $\delta =\frac{\eps}{2d} \leq \frac{1}{6}$,
and $m \geq 2d/\eps \geq 2d \geq 6$.  Then a simple calculation shows
that 
\begin{equation} \label{E12}
a = \frac{1}{1-\delta} \le 1+ \frac{6}{5} \delta = 1+ \frac{3\eps}{5d}. 
\end{equation} 

It is also clear that $a =\frac{1}{1-\delta} > 1+ \delta$.
So $a$ satisfies
\begin{equation} \label{E13}
1 < 1+ \delta <
a =\frac{1}{1-\delta} \leq 1+ \frac{6}{5} \delta \leq \frac{6}{5}. 
\end{equation} 
Since $n \geq 12$ and $a \leq \frac{6}{5}$, it follows from the
second inequality in (3) that $k \geq 15$. 
We now derive an upper bound on $k$ as a function of $n$, $d$ and
$\eps$. First observe that 
$$ \log{a} =\log{\frac{1}{1-\delta} } \geq \log(1+\delta). $$
We also have 
$$ \ln(1+\delta) \geq 0.9 \delta {\rm \ \ for \ }
\delta \leq \frac{1}{6}. $$
From~\eqref{E28} we deduce the following sequence of inequalities: 
\begin{equation} \label{E5}
 k -1 \leq \frac{\log(n+1)}{\log{a}} \leq 
\frac{\log(n+1)}{\log(1+\delta)} =
\frac{\log(n+1)\cdot \ln{2}}{\ln(1+\delta)} \leq 
\frac{\log(n+1)\cdot \ln{2}}{0.9 \delta} \leq
\frac{0.78 \log(n+1)}{\delta}.
\end{equation} 

From~\eqref{E5}, a straightforward calculation 
(where we use $n \geq 12$ and $\delta \leq 1/6$) gives
\begin{equation} \label{E20}
k \leq \frac{0.78 \log(n+1)}{\delta} +1 \leq 
\frac{0.78 \log(n+1) + 1/6}{\delta}  \leq 
\frac{\log{n}}{\delta}= \frac{2d}{\eps} \cdot \log{n}.
\end{equation} 

\paragraph{Overview of the algorithm.}
%We first give an overview of the algorithm. 
By a direct generalization of Observation~\ref{O1}, we can assume w.l.o.g.\  
that $R=[0,1]^d$. Let $S$ be the set of $n$ points contained in $R$.
The algorithm generates a finite set $\B$ of {\em canonical boxes};
to be precise, only a subset of {\em large} canonical boxes.
For each large canonical box $B_0 \in \B$, a corresponding {\em
canonical grid} is considered, and $B_0$ is placed with its lowest
corner at each such grid position and tested for emptiness and containment in $R$. 
The one with the largest volume amongst these is returned in the end. 

\paragraph{Canonical boxes and their associated grids.}
%We now describe the canonical boxes and their associated grids.
Consider the set $\B$ of {\em canonical boxes}, whose 
all side lengths are elements of
\begin{equation} \label{E18}
\X=\left\{\frac{a^i}{a^{k+1}}, i=0,1,\ldots,k-1\right\}. 
\end{equation} 
For a given canonical box $B_0 \in \B$, with sides 
$X_1,\ldots,X_d \in \X$,
consider the {\em canonical grid associated with $B_0$} 
with points of coordinates 

\begin{equation} \label{E10}
\left(\frac{i_1 X_1}{m}, \ldots, \frac{i_d X_d}{m}\right), \ \ 
i_1,\ldots, i_d \geq 0 
\end{equation} 
contained in $U_d$. 

Let $B$ be a maximum-volume empty box in $R=U_d$, with
$V_{\rm max}=\vol(B)$. By the trivial inequality $A_d(n) \geq
\frac{1}{n+1}$ of Proposition~\ref{P2}, we have
$V_{\rm max} \geq \frac{1}{n+1}$. This lower bound is crucial
in the design of our approximation algorithm, as it enables us to
bound from above the number of large canonical boxes (canonical boxes
of smaller volume can be safely ignored).  

Consider the following set $\I$ of $k+1$ intervals 
\begin{equation} \label{E17}
\I=\left\{\left[\frac{a^i}{a^{k+1}}, \frac{a^{i+1}}{a^{k+1}}\right),
  i=0,1,\ldots,k\right\}. 
\end{equation} 
Observe that for each $i=1,\ldots,d$, the extent in the $i$th
coordinate of $B$ is at least $\frac{a}{a^{k+1}}=\frac{1}{a^k}$,
since otherwise we would have 
$\vol(B) < \frac{1}{a^k} < \frac{1}{n+1}$, a contradiction.
Let $Z_i$ be the extent in the $i$th coordinate of $B$, for
$i=1,\ldots,d$. By the above observation, for each $i=1,\ldots,d$,
$Z_i$ belongs to one of the last $k$ intervals in the set $\I$. That
is, there exists an integer $y_i \in \{0,1,\ldots,k-1\}$, such that 
\begin{equation} \label{E1}
Z_i \in 
\left[\frac{a^{y_i+1}}{a^{k+1}}, \frac{a^{y_i+2}}{a^{k+1}}\right). 
\end{equation} 

The next lemma shows that $B$ contains an (empty) canonical box
with side lengths 
\begin{equation} \label{E2}
X_i= \frac{a^{y_i}}{a^{k+1}}, \ i=1,\ldots,d, 
\end{equation} 
at some position in the canonical grid associated with it. 
We call such a canonical box contained in a maximum-volume empty box,
a  {\em large canonical box}. Two key properties of large
canonical boxes are proved in Lemma~\ref{L9} and Lemma~\ref{L8}.

\begin{lemma} \label{L7}
If for each $i=1,\ldots,d$, 
the extent in the $i$th coordinate of $B$ belongs to the interval
as in \eqref{E1}, then $B$ contains an (empty) large canonical box
$B_0$ with side lengths as in \eqref{E2} at some position in the
canonical grid associated with it.    
\end{lemma}
\begin{proof}
It is enough to prove the containment for each coordinate axis $i$. 
Let $I$ and $I_0$ be the corresponding intervals of $B$
and $B_0$, respectively.
Assume for contradiction that the placement of $I_0$ with its left end
point at the first canonical grid position in $I$ is not contained in 
$I$. But then we have, by taking into account the grid cell lengths:
$$ \frac{a^{y_i+1}}{a^{k+1}} \leq |I| <
|I_0| + \frac{|I_0|}{m} \leq |I_0| + \delta \cdot |I_0| =
\left(1+\delta \right) \frac{a^{y_i}}{a^{k+1}}, $$
and consequently,
$$ a < 1+\delta. $$
We reached a contradiction to the 2nd inequality in \eqref{E13}, and
the proof is complete. 
\end{proof}

We now show that the (empty) large canonical box $B_0 \subset B$ from the
previous lemma yields a $(1-\eps)$-approximation for the empty box $B$
of maximum volume.

\begin{lemma} \label{L9}
$$ \vol(B_0) \geq (1-\eps) \cdot \vol(B). $$
\end{lemma}
\begin{proof}
By \eqref{E2} and \eqref{E1}, 
$$ \vol(B_0) = \prod_{i=1}^d \frac{a^{y_i}}{a^{k+1}} =
\frac{1}{a^{2d}} \prod_{i=1}^d \frac{a^{y_i+2}}{a^{k+1}} \geq
\frac{1}{a^{2d}} \cdot \vol(B). $$
It remains to be shown that 
$$ \frac{1}{a^{2d}} \geq 1-\eps. $$
But this follows from our choice of $a$ and from Bernoulli's
inequality: 
$$ (1+x)^q \geq 1+qx, {\rm \ for \ any \ } x \geq -1, 
{\rm \ and \ any \ positive \ integer \ } q. $$
Indeed,
$$ \frac{1}{a^{2d}} = \left(1-\frac{\eps}{2d}\right)^{2d} \geq
1- 2d \cdot \frac{\eps}{2d}= 1-\eps, $$
and the proof of Lemma~\ref{L9} is complete.
\end{proof}

Observe that the number of canonical boxes in $\B$ is exactly $k^d$,
and by~\eqref{E20} is bounded from above as follows:
\begin{equation} \label{E16}
k^d \leq \left(\frac{2d}{\eps} \right)^d \cdot \log^d{n}.
\end{equation}
We can prove however a better upper bound on the number of large canonical
boxes. 

\begin{lemma} \label{L8}
The number of large canonical boxes in $\B$ is at most
$$ \left(\frac{2e}{\eps}\right)^d \cdot \log^d{n}. $$
\end{lemma}
\begin{proof}
Recall that $\vol(B)$ satisfies
$$ \frac{1}{a^k} < \frac{1}{n+1} \leq \vol(B) \leq
\prod_{i=1}^d \frac{a^{y_i+2}}{a^{k+1}} =
\frac{a^{2d} \prod_{i=1}^d a^{y_i}}{a^{dk+d}}=
\frac{a^{d} \prod_{i=1}^d a^{y_i}}{a^{dk}}, $$
for some integers $y_i \in \{0,1,\ldots,k-1\}$. 
It follows immediately that
\begin{equation} \label{E6}
dk-k-d \leq \sum_{i=1}^d y_i \leq dk-d, 
\end{equation} 
and we want an upper bound on the number of solutions of  \eqref{E6}.
Make the substitution $z_i=k-1-y_i$, for $i=1,2,\ldots,d$. 
So $z_i \in \{0,1,\ldots,k-1\}$, for $i=1,2,\ldots,d$. 
The above inequalities are equivalent to
\begin{equation} \label{E29}
0 \leq \sum_{i=1}^d z_i \leq k. 
\end{equation} 

Let $t$ be a nonnegative integer. It is well-known (see for instance 
\cite{T95}) that the number of
nonnegative integer solutions of the equation $ \sum_{i=1}^d z_i =t$ 
equals ${t+d-1 \choose d-1}$, that is, when we ignore the upper bound
on each $z_i$. Summing over all values of $t \in \{0,1,\ldots,k\}$,
and using a well-known binomial identity (see for instance 
\cite[p. 217]{T95}) 
yields that the number of solutions of \eqref{E29}, hence also of \eqref{E6},
is no more than
$$ \sum_{t=0}^{k} {t+d-1 \choose d-1} = {k+d-1+1 \choose d-1+1} =
{k+d \choose d} . $$
A well-known upper bound approximation for binomial coefficients
$$ {n \choose k} \leq \left(\frac{en}{k}\right)^k, $$
for positive integers $n$ and $k$ with $1 \leq k \leq n$, further yields that
\begin{equation} \label{E7}
{k+d \choose d} \leq \left(\frac{e(k+d)}{d}\right)^d =
e^d \left(\frac{k+d}{d}\right)^d. 
\end{equation} 

We now check that 
$$ k+d \leq \frac{\log{n}}{\delta}. $$
Recall inequality \eqref{E5}. A straightforward calculation 
(where we use $n \geq 12$, $d \ge 3$, and $\eps \leq 1$), gives
\begin{equation} \label{E15}
k+d \leq \frac{0.78 \log(n+1)}{\delta} + d+1 =
\frac{0.78 \log(n+1)+\frac{d+1}{2d}\eps}{\delta} \leq
\frac{\log{n}}{\delta}= \frac{2d}{\eps} \cdot \log{n},
\end{equation} 
as claimed. Substituting this upper bound into \eqref{E7} yields
\begin{equation} \label{E9}
{k+d \choose d} \leq 
e^d \left(\frac{2d}{d \eps}\cdot \log{n}\right)^d=
\left(\frac{2e}{\eps}\right)^d \cdot \log^d{n}, 
\end{equation} 
as required. This expression is an upper bound on the number of 
solutions of \eqref{E6}, hence also on the number of large canonical
boxes in $\B$.  
\end{proof}

Given a grid with cell lengths $x_1,x_2,\ldots,x_d$, we superimpose it so that the
origin of $U_d$ is a grid point of the above grid. Denote the
corresponding grid cells by index tuples $(i_1,i_2,\ldots,i_d)$, 
where $i_1,i_2,\ldots,i_d \geq 0$. 
Note that some of the grid cells on the boundary of $U_d$ may be
smaller. Given a grid $G$ superimposed on $U_d$, let $M(G)$ be the
number of cells (with nonempty interior) into which $U_d$ is partitioned. 

Consider a (fixed) canonical box, say $B_0$,  with side lengths as in
\eqref{E2}. The associated canonical grid, say $G_0$, has side lengths
$m$ times smaller in each coordinate. 
We now derive an upper bound on the number of canonical
grid positions where a canonical box is placed and tested for emptiness,
according to \eqref{E10}. 

\begin{lemma} \label{L10}
The number of canonical grid positions for placing $B_0$ in $G_0$ is
bounded as follows: 
$$ M(G_0) \leq 12 \cdot \left( \frac{2d}{\eps} \right)^d \cdot n. $$
\end{lemma}
\begin{proof}
We have
$$ M(G_0) \leq 
\prod_{i=1}^d \left \lceil \frac{m \cdot a^{k+1}}{a^{y_i}} \right \rceil \leq
\prod_{i=1}^d \left(\frac{m \cdot a^{k+1}}{a^{y_i}} +1 \right). 
$$
Observe that
$$ \frac{m \cdot a^{k+1}}{a^{y_i}} +1 = 
\frac{m \cdot a^{k+1}+a^{y_i} }{a^{y_i}} \leq 
\frac{m \cdot a^{k+1}+a^{k-1} }{a^{y_i}} \leq 
\frac{(m+1) a^{k+1}}{a^{y_i}}. 
$$
By substituting this bound in the product we get that
\begin{align} \label{E11}
M(G_0) &\leq \prod_{i=1}^d \frac{(m+1) \cdot a^{k+1}}{a^{y_i}} =
(m+1)^d \prod_{i=1}^d \frac{a^{k+1}}{a^{y_i}} = 
(m+1)^d \cdot \frac{a^{kd+d}}{\prod_{i=1}^d a^{y_i}} \nonumber \\
&\leq (m+1)^d \cdot \frac{a^{kd+d}}{a^{kd-k-d}} = 
(m+1)^d \cdot a^{2d} \cdot a^k. 
\end{align} 
For the last inequality above we used \eqref{E6}.
We now bound from above each of the three factors in \eqref{E11}.
For bounding the second and the third factors we use inequalities
\eqref{E12} and \eqref{E28}, respectively.
$$ (m+1)^d = \left( \left \lceil \frac{2d}{\eps} \right \rceil+1\right) ^d \leq
\left( \frac{2d}{\eps} +2 \right)^d = 
\left( \frac{2d}{\eps} \right)^d \left( 1+ \frac{\eps}{d} \right)^d \leq
\left( \frac{2d}{\eps} \right)^d \cdot e^{\eps} \leq
e \left( \frac{2d}{\eps} \right)^d. $$
$$  a^{2d} \leq \left(1+ \frac{3\eps}{5d} \right)^{2d} \leq
e^{6\eps/5} \leq e^{6/5}. $$
$$ a^k = a \cdot a^{k-1} \leq  a(n+1) \leq \frac{6}{5} \cdot (n+1)
\leq \frac{13n}{10}, {\rm \ for \ } n \geq 12. $$
Substituting these upper bounds in \eqref{E11} gives the desired bound:
$$ M(G_0) \leq e^{11/5} \left( \frac{2d}{\eps} \right)^d \cdot \frac{13n}{10} \leq
12 \cdot \left( \frac{2d}{\eps} \right)^d \cdot n. 
$$

\vspace{-2\baselineskip}
\end{proof}

%\vspace{-0.5\baselineskip}
\paragraph{Testing canonical boxes for emptiness.}

Given a grid with cell lengths $x_1,x_2,\ldots,x_d$, denote the corresponding 
grid {\em cell counts} or {\em cell numbers}
(i.e., the number of points) in cell
$(i_1,i_2,\ldots,i_d)$ by $n(i_1,i_2,\ldots,i_d)$.
For simplicity, we can assume w.l.o.g.\ that in all the grids that are
generated by the algorithm, no point of $S$ lies on a grid cell
boundary. Indeed the points of $S$ on the boundary of $R=U_d$ can be
safely ignored, and the above condition holds with probability $1$
if instead of the given $\eps$, the algorithm uses a value chosen
uniformly at random from the interval $[(1-\frac{1}{2d})\eps, \ \eps]$;
see also the setting of the parameters in \eqref{E14}. 
Given a grid $G$, and dimensions (array sizes) $M_1,\ldots,M_d \geq 1$, 
a floating box at some position aligned with it, that is,
whose lower left corner is a grid point, and with the specified 
dimensions is called a {\em grid box}. All the canonical boxes
generated by our algorithm are in fact grid boxes. 

The next four lemmas (\ref{L3},~\ref{L4},~\ref{L5},~\ref{L6})
outline the method we use for efficiently computing the 
number of points in $S$ in a rectangular box, over a sequence 
of boxes. In particular these boxes can be tested for emptiness within
the same specified time. 

\begin{lemma} \label{L3}
Let $G$ be a grid with cell lengths $x_1,x_2,\ldots,x_d$, 
superimposed on $U_d$, with $M(G)$ cells.
Then the number of points of $S$ lying in each cell, over all cells,
can be computed in $O(d \cdot n + M(G))$ time. 
\end{lemma}
\begin{proof}
The number of points in each cell of $M(G))$ is initialized to $0$. 
For each point $p \in S$, its cell index tuple (label) is computed
in $O(d)$ time using the floor function for each coordinate, 
and the corresponding cell count is updated. 
\end{proof}

\noindent{\bf Remark.} If the floor function is not an option,
the number of points in each cell can be computed using binary search 
for each coordinate. The resulting time complexity is
$O(n \cdot \log{M(G)})$. 

\medskip
Denote by $N(i_1,i_2,\ldots,i_d)$ the number of points in $S$ in the
box with lower left cell $(0,0,\ldots,0)$, and upper right cell
$(i_1,i_2,\ldots,i_d)$; we refer to the numbers $N(i_1,i_2,\ldots,i_d)$ 
as {\em corner box} numbers.

\begin{lemma} \label{L4}
Given a grid $G$ with cell lengths $x_1,x_2,\ldots,x_d$ placed at the origin, 
with $M(G)$ cells, and grid cell counts $n(i_1,i_2,\ldots,i_d)$,  over all
cells, the corner box numbers $N(i_1,i_2,\ldots,i_d)$, over all cells,
can be computed in $O(2^d \cdot M(G))$ time.  
\end{lemma}
\begin{proof}
Define $N(i_1,i_2,\ldots,i_d)=0$, if $i_j<0$ for some $j$. 
Let $b=(b_1,b_2,\ldots,b_d) \in \{0,1\}^d$ be a binary vector.
Let the {\em parity} of $b$ be 
$\pi(b)=b_1 \oplus b_2 \oplus \cdots \oplus b_d$. 
By the inclusion-exclusion principle, the corner box numbers
are given by the following formula with at most $2^d$ operations:
$$ N(i_1,i_2,\ldots,i_d)= n(i_1,i_2,\ldots,i_d)+
\sum_{\substack
{b=(b_1,b_2,\ldots,b_d)\\ b \neq (0,0,\ldots,0)}} 
(-1)^{\pi(b)+1} N(i_1-b_1,i_2-b_2,\ldots,i_d-b_d). $$
Since $G$ has $M(G)$ cells, the bound follows.
\end{proof}

\begin{lemma} \label{L5}
Given is a grid $G$ with cell lengths $x_1,x_2,\ldots,x_d$ placed at the origin, 
with $M(G)$ cells, and corner box numbers $N(i'_1,i'_2,\ldots,i'_d)$,
over all cells. Let $B_0$ be a (canonical) grid box with dimensions
(array sizes) $M_1,\ldots,M_d \geq 1$, and lower left cell
$(i_1,i_2,\ldots,i_d)$. Then the number of points of $S$ in $B_0$,
denoted $N(B_0)$, can be computed in $O(2^d)$ time. 
\end{lemma}
\begin{proof}
Let $j_1=i_1+M_1-1, \ldots, j_d=i_d+M_d-1$ be the upper right cell of
$B_0$. By the inclusion-exclusion principle, the corner box number
$N(j_1,j_2,\ldots,j_d)$ can be computed as follows:
$$ N(j_1,j_2,\ldots,j_d)= N(B_0)+
\sum_{\substack
{b=(b_1,b_2,\ldots,b_d)\\ b \neq (0,0,\ldots,0)}} 
(-1)^{\pi(b)+1} N(j_1-b_1 M_1,j_2-b_2 M_2,\ldots,j_d-b_d M_d). $$
Hence $N(B_0)$ can be extracted from the above formula with at most $2^d$
operations. 
\end{proof}

Let $Q(i_1,i_2,\ldots,i_d)$ be the number of points in $S$ in the
canonical box of dimensions (array sizes) $M_1,\ldots,M_d \geq 1$, 
and lower left cell $(i_1,i_2,\ldots,i_d)$.

\begin{lemma} \label{L6}
Given is a grid $G$ with cell lengths $x_1,x_2,\ldots,x_d$ placed at the origin, 
with $M(G)$ cells, and corner box numbers $N(i'_1,i'_2,\ldots,i'_d)$,
over all cells. Then the numbers (counts)  $Q(i_1,i_2,\ldots,i_d)$,
over all cells, can be computed in $O(2^d \cdot M(G))$ time. 
\end{lemma}
\begin{proof}
There are $M(G)$ cells determined by $G$ in $U_d$, and for each, apply
the bound of Lemma~\ref{L5}.
\end{proof}

%\vspace{-0.5\baselineskip}
\paragraph{The last step in the proof of Theorem~\ref{T4}.}

For each canonical box, say $B_0$, there is a unique associated
canonical grid, say $G_0$. The time taken to test $B_0$ for emptiness
and containment in $R$ when placed at all grid positions in $G_0$, is
obtained by adding the running times in lemmas~\ref{L3}, 
\ref{L4}, and~\ref{L6}:
\begin{equation} \label{E24}
O(d \cdot n + M(G_0) + 2^d \cdot M(G_0))=
O\left(2^d \cdot \left( \frac{2d}{\eps} \right)^d \cdot n \right)=
O\left(\left( \frac{4d}{\eps} \right)^d \cdot n \right), 
\end{equation} 
where we have used the upper bound on $M(G_0)$ in Lemma~\ref{L10}.
By multiplying this with the upper bound on the number of large canonical
boxes in Lemma~\ref{L8}, we get that the total running time of the 
approximation algorithm is 
\begin{equation} \label{E25}
O\left(\left(\frac{2e}{\eps}\right)^d \cdot \log^d{n}
\cdot \left( \frac{4d}{\eps} \right)^d \cdot n \right)=
O\left(\left( \frac{8 e d}{\eps^2} \right)^d \cdot n \cdot
\log^d{n} \right). 
\end{equation} 
The proof of Theorem~\ref{T4} is now complete.

\later{
\section{Concluding remarks}

Reducing the gap between the lower and upper bounds on $A_d(n)$,
particularly in higher dimensions remains an interesting problem. 
Other questions can be asked regarding the computational complexity
of computing a maximum-volume empty box. 
We list some specific questions and directions for further study:

\begin{itemize}
\itemsep 0.01in
\item [(1)] Is the dependence on $d$ necessary in the upper bound on
  $A_d(n)$ as given by our Theorem~\ref{T3}, or is $A_d(n) \leq
  \frac{C}{n}$, where $C$ is an absolute constant? As a preliminary
  question: Given $d$ points in the unit hypercube $[0,1]^d$, is there
  always an empty box of volume $C$, where $C$ is an absolute
  constant, or does $A_d(d)$ tend to zero with the dimension? 
\item [(2)] Most likely the dependence on $n$ of the running time of
  our approximation algorithms, for boxes and hypercubes, is close to
optimal. However, reducing the dependence on $d$ and $\eps$ in the
running time may extend the range of dimensions for which the
algorithm is practical.  
\end{itemize}
} % later

% ****************************************************************

%\newpage
\appendix

\section{Proof of Lemma~\ref{L2}} \label{app:L2}

To see that $A(4) \leq \frac{1}{4}$, consider the $4$ points
$(\frac{1}{4},\frac{1}{2})$, 
$(\frac{1}{2},\frac{1}{4})$,
$(\frac{1}{2},\frac{3}{4})$,
$(\frac{3}{4},\frac{1}{2})$,
%$(1/4,1/2)$, $(1/2,1/4)$, $(1/2,3/4)$, $(3/4,1/2)$, 
and check that the largest empty rectangle has area  $\frac{1}{4}$. 
Next we prove the lower bound. 
Let $S=\{p_1,p_2,p_3,p_4\}$ be a set of $4$ points, and assume without
loss of generality that they are in lexicographic order: 
$x(p_1) \leq x(p_2) \leq x(p_3) \leq x(p_4)$,  
and if $x(p_i) = x(p_j)$ for $i<j$, then $y(p_i)< y(p_j)$. 
We can also assume that $y(p_1) \leq y(p_2)$. Encode each possible such
$4$-point configuration by a permutation $\pi$ of $1,2,3,4$ as follows:
for $i<j$, $\pi(i) <\pi(j)$ if and only if $y(p_i) \leq y(p_j)$. 
For example $\pi=(2,4,3,1)$ encodes the configuration shown in
Fig.~\ref{f1}(right). 
\begin{figure} [htb]
\centerline{\epsfxsize=4in \epsffile{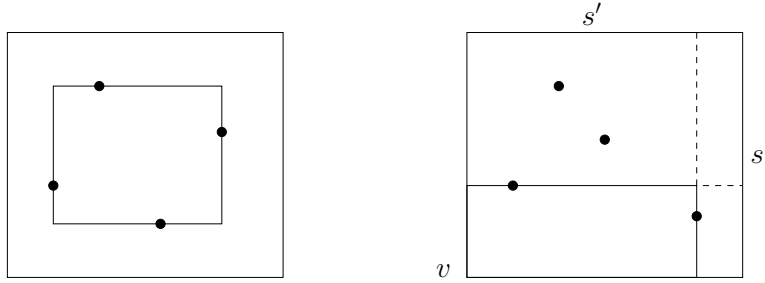}}
\caption{\small Left: $\pi=(2,4,1,3)$ is special.
Right: $\pi=(2,4,3,1)$ is non-special; $s$ is the right side of $U$, $v$
is the lower left corner of $U$, and $s'$ is the top side of $U$.} 
\label{f1}
\end{figure}

By our assumption $y(p_1) \leq y(p_2)$, there are only $12$
permutations (types) out of the total of $4! = 24$ to consider, those with
$\pi(1)<\pi(2)$. 
Two of these permutations, namely $(2,4,1,3)$ and $(3,4,1,2)$, are
called {\em special}: the $4$ points are in convex position and there
is an empty rectangle $R \subset U$, with one of these points on each side of
$R$. All the remaining $10$ permutations are called {\em non-special}.
We distinguish two cases: 

\medskip
{\em Case 1:} $S$ is encoded by a special permutation.
For each of the four sides $s$ of $U$, let $P(s)$ be the largest
empty rectangle containing $s$.
See Fig.~\ref{f1}(left) for an example. 
We can assume that the area of each rectangle 
$P(s)$ is smaller than $\frac{1}{4}$, since else we are done. But then it
follows that each of the four sides of $R$ is longer than 
$1-\frac{2}{4}=\frac{1}{2}$, so the area
of $R$ is larger than $\frac{1}{2} \cdot \frac{1}{2} = \frac{1}{4}$,
so this case is settled.

\medskip
{\em Case 2:} $S$ is encoded by a non-special permutation.
For each of the four vertices $v$ of $U$, let $Q(v)$ be the largest
empty rectangle having $v$ as a vertex. A routine verification shows that for
each of the $10$ non-special permutations there is a side $s$ of $U$
and a vertex $v$ of $U$ such that (i) $P(s)$ and $Q(v)$ have a common
boundary segment, and (ii) $v$ is an endpoint of the side opposite to $s$.
More precisely,
if $\pi$ is one of six permutations
$(1,2,3,4)$, $(1,2,4,3)$, $(1,3,2,4)$, 
$(1,3,4,2)$, $(1,4,2,3)$, $(1,4,3,2)$, 
then $s$ is the left side, and $v$ is the lower-right corner;
if $\pi$ is one of four permutations
$(2,3,1,4)$, $(2,3,4,1)$, $(2,4,3,1)$, $(3,4,2,1)$, 
then $s$ is the right side, and $v$ is the lower-left corner. 
See Fig.~\ref{f1}(right) for an example. 

As in Case 1, we can assume that
the area of $P(s)$ is smaller than $\frac{1}{4}$, thus its shorter side
is smaller than $\frac{1}{4}$. By the same token, one of the sides of $Q(v)$ 
is longer than $1-\frac{1}{4}=\frac{3}{4}$, hence the other side must
be shorter than $\frac{1}{3}$, 
since otherwise the area of $Q(v)$ would exceed $\frac{1}{4}$.
Let $s'$ be the side of $U$ adjacent to $s$ and disjoint from $v$.
Consequently, the rectangle $R'$ with side $s'$ and adjacent to $Q(v)$
has the other side longer than $1-\frac{1}{3}=\frac{2}{3}$. Observe
that $R'$ has at most two points in its interior. By Lemma~\ref{L1}
and Observation~\ref{O1}, $R'$ contains an empty rectangle of area at
least 
$$ \frac{2}{3} \cdot \xi =\frac{3-\sqrt{5}}{3} = 0.254\ldots > \frac{1}{4}, $$
as required. This concludes the analysis of the second case.

Thus in both cases, there is an empty rectangle of area at least $\frac{1}{4}$.
\qed

\section{Empty squares and hypercubes}\label{sec:square}

Define $A'_d(n)$ as the volume of the largest empty axis-parallel
hypercube (over all $n$-element point sets in in $[0,1]^d$), analogous
to $A_d(n)$ for the largest empty axis-parallel box.  
For simplicity we sometimes omit the subscript $d$ in the planar case
($d=2$).  That is, $A'(n)$ denotes the area of the largest empty 
axis-parallel square. Then for any fixed dimension $d$, our next
theorem shows that $A'_d(n) =\Theta\left(\frac{1}{n}\right)$, too: 

\begin{theorem} \label{T5}
For a fixed $d$, $A'_d(n) =\Theta\left(\frac{1}{n}\right)$. 
More precisely,
\begin{equation} \label{E31}
\frac{1}{(n^{1/d} + 1)^d} \leq A'_d(n) \leq
\frac{1}{(\lfloor n^{1/d} \rfloor + 1)^d}. 
\end{equation}
\end{theorem}
\begin{proof}
We will prove the bounds for the planar case $d = 2$:
$$ \frac{1}{(\sqrt{n} + 1)^2} \leq A'(n) \leq
\frac{1}{(\lfloor \sqrt{n} \rfloor + 1)^2}. $$
The proof can be easily generalized for $d \geq 3$.

We first prove the lower bound. Let $S$ be a set of $n$ points in the unit
square $U$. Let $x$ be a positive number to be determined. Let $X$ be
an axis-parallel square of side $1 - x$ that is concentric with $U$. For
each point $p \in S$, place an axis-parallel (open) square of side $x$
centered at $p$. If there is a point $q \in X$ that is not covered by the
union of the $n$ squares, then the axis-parallel square of side $x$
centered at $q$ is an empty square contained in $U$.

The area of $X$ is $(1 - x)^2$. The total area of $n$ squares of side
$x$ is $n x^2$. Let $x$ be the solution to the following equation
$$ (1 - x)^2 = n x^2. $$
The solution is $x = \frac{1}{\sqrt{n} + 1} $. For this value of $x$, 
either the $n$ small squares cover $X$ with no interior overlap among
themselves, or there is interior overlap and they don't cover $X$. In
either case, there exists an open axis-parallel square of side length
$x$, centered at a point in $X$, and empty of points in $S$. Consequently,  
$$ A'(n) \geq x^2 = \frac{1}{(\sqrt{n} + 1)^2}. $$

We next prove the upper bound. Let $k = \lfloor \sqrt{n} \rfloor$. Note
that $n \ge k^2$. Partition the unit square $U$ into a $(k+1) \times (k+1)$
square grid of cell length $1/(k+1)$. Place a point at each of the $k^2$
grid vertices in the interior of $U$. Then any axis-parallel square
contained in $U$ whose side is longer than $1/(k+1)$, must be
non-empty. Consequently, 
$$ A'(n) \leq \frac{1}{(\lfloor \sqrt{n}\rfloor  + 1)^2}. $$

It remains to show that \eqref{E31} implies that for a fixed $d$, 
we have $A'_d(n) =\Theta\left(\frac{1}{n}\right)$.
The following inequalities are straightforward:
$$ n^{1/d} + 1 \leq 2 n^{1/d} \quad\Rightarrow\quad 
(n^{1/d} + 1)^d \leq 2^d n, $$
$$ n^{1/d} \leq \lfloor n^{1/d} \rfloor +1 \quad\Rightarrow\quad
n \leq (\lfloor n^{1/d} \rfloor + 1)^d. $$
Putting them together yields
$$ \frac{1}{2^d n} \leq \frac{1}{(n^{1/d} + 1)^d} \leq A'_d(n) \leq
\frac{1}{(\lfloor n^{1/d} \rfloor + 1)^d} \leq 
\frac{1}{n}, $$
as claimed.
\end{proof}

\section{A $(1-\eps)$-approximation algorithm for finding the largest
empty \\ hypercube}\label{sec:approx2}  

Let $R$ be an axis-parallel $d$-dimensional hypercube in $\RR^d$ 
containing $n$ points. In this section, we present an efficient
$(1-\eps)$-approximation algorithm for computing a maximum-volume
empty axis-parallel hypercube contained in $R$. 
\later{
Recall that 
with exact algorithms, a largest empty hypercube can be found faster
than a largest empty box~\cite{BK09a,BK09b}, as mentioned in the
introduction. In our case, with approximation algorithms, 
the situation is analogous, and we are able to obtain a faster
algorithm for finding the largest hypercube:
} % later

\begin{theorem} \label{T6}
Given an axis-parallel $d$-dimensional hypercube $R$ in $\RR^d$ containing
$n$ points, there is a $(1-\eps)$-approximation algorithm, running in  
$$ O\left( \frac{d^2}{\eps} \cdot n \log{n} +
\left( \frac{4d}{\eps} \right)^{d+1} \cdot n^{1/d} \log{n}  \right)
$$
time, for computing a maximum-volume empty axis-parallel hypercube 
contained in $R$. 
\end{theorem}
\begin{proof} The overall structure of the algorithm is
similar to that for finding the largest empty box. 
We can assume w.l.o.g. that $R=U_d=[0,1]^d$, $n \geq 12$, and $d \geq 3$. 
Recall that, by Theorem~\ref{T5}, the volume of a largest empty
hypercube in $U_d$ is at least $(n^{1/d} + 1)^{-d}$.  
We set the parameters $\delta$, $m$ and $a$ as in equation \eqref{E14}. 
Inequalities \eqref{E12} and \eqref{E13} also follow.
Let now $k$ be the unique positive integer such that
\begin{equation} \label{E30}
a^{k-1} \leq n^{1/d} + 1 <a^k.
\end{equation} 
Thus 
$$ a^{k-1} \leq n^{1/d} + 1 \leq 2n^{1/d}. $$
Since $n \geq 12$ and $d \geq 3$ we have 
$$ k-1 \leq \frac{1+ \frac{1}{d} \log{n}}{\log a} \leq
\frac{(\frac{1}{3}+ \frac{1}{d}) \log{n}}{\log a} \leq
\frac{2}{3} \cdot \frac{\log{n}}{\log a} \leq 
\frac{2\log{n}}{3} \cdot \frac{\ln{2}}{0.9 \delta} \leq
\frac{3\log{n}}{5\delta}. $$
It follows that 
\begin{equation} \label{E27}
k \leq 1 + \frac{3\log{n}}{5\delta} \leq \frac{\log{n}}{\delta} =
\frac{2d}{\eps} \cdot \log{n}. 
\end{equation} 

Consider the set $\H$ of $k$ {\em canonical hypercubes} whose sides
are elements of $\X$ (as in \eqref{E18}):
\begin{equation} \label{E26}
\X=\left\{\frac{a^i}{a^{k+1}}, i=0,1,\ldots,k-1\right\}. 
\end{equation} 

For a given canonical hypercube $H_0 \in \H$, with side $X \in \X$,
consider the {\em canonical grid associated with $H_0$} 
with points of coordinates 
\begin{equation} \label{E19}
\left(\frac{i_1 X}{m}, \ldots, \frac{i_d X}{m}\right), \ \ 
i_1,\ldots, i_d \geq 0 
\end{equation} 
contained in $U_d$. 

Consider the set $\I$ of $k+1$ intervals (as in \eqref{E17}): 
\begin{equation} \label{E21}
\I=\left\{\left[\frac{a^i}{a^{k+1}}, \frac{a^{i+1}}{a^{k+1}}\right),
  i=0,1,\ldots,k\right\}. 
\end{equation} 

Let $H$ be a maximum-volume empty hypercube in $R=U_d$, with side
length $Z$ and $V_{\rm max}=\vol(H)$. Observe that $Z \geq \frac{a}{a^{k+1}}$: 
indeed, $Z <\frac{a}{a^{k+1}}$ would imply that 
$$ Z^d < \frac{1}{a^{kd}} < \frac{1}{(n^{1/d}+1)^d}, $$
in contradiction to the lower bound in Theorem~\ref{T5}. 
This means that $Z$ belongs to one of the last $k$ intervals in the
set $\I$. That is, there exists an integer $y \in \{0,1,\ldots,k-1\}$, such that 
\begin{equation} \label{E22}
Z \in \left[\frac{a^{y+1}}{a^{k+1}}, \frac{a^{y+2}}{a^{k+1}}\right). 
\end{equation} 

Analogous to Lemma~\ref{L7}, we conclude that $H$ contains a 
{\em large canonical hypercube}, say $H_0$, whose side is
\begin{equation} \label{E23}
X= \frac{a^{y}}{a^{k+1}}, 
\end{equation} 
at some position in the canonical grid associated with it. 
Analogous to Lemma~\ref{L9}, we show that 
$\vol(H_0) \geq (1-\eps) \cdot \vol(H)$:
By \eqref{E23} and \eqref{E22}, 
$$ \vol(H_0) = \left(\frac{a^{y}}{a^{k+1}}\right)^d=
\frac{1}{a^{2d}} \cdot \left(\frac{a^{y+2}}{a^{k+1}} \right)^d \geq
\frac{1}{a^{2d}} \cdot \vol(H) \geq (1-\eps) \cdot \vol(H),
$$
since the setting of $a$ is the same as before. 
Analogous to Lemma~\ref{L8}, now \eqref{E27} is the upper bound we
need on the number of canonical hypercubes. 
The bound in Lemma~\ref{L10} needs to be adjusted because $k$ is
chosen differently, and we have a different upper bound on the third
factor in the product, $a^k$. From the definition of $k$ in
\eqref{E30} and from \eqref{E13} we deduce 
$$ a^k =a \cdot a^{k-1} \leq a (n^{1/d} +1) \leq 
2a n^{1/d} \leq \frac{12}{5} n^{1/d}. $$
The resulting bound analogous to that in Lemma~\ref{L10} is now
\begin{equation} \label{E32}
M(G_0) \leq e^{11/5} \left( \frac{2d}{\eps} \right)^d \cdot 
\frac{12}{5} n^{1/d} \leq
22 \cdot \left( \frac{2d}{\eps} \right)^d \cdot n^{1/d}. 
\end{equation} 

The time taken to test $H_0$ for emptiness and containment in $R$ when
placed at all relevant grid positions is now
\begin{equation*} \label{E33}
O(d \cdot n + M(G_0) + 2^d \cdot M(G_0))=
O\left( dn + 2^d \cdot \left( \frac{2d}{\eps} \right)^d \cdot n^{1/d} \right)=
O\left( dn + \left( \frac{4d}{\eps} \right)^d \cdot n^{1/d} \right). 
\end{equation*} 

By multiplying this with the upper bound in~\eqref{E27}, 
on the number of canonical hypercubes, we get that the total
running time of the approximation algorithm is 
$$ O\left( \frac{d^2}{\eps} \cdot n \log{n} +
\left( \frac{4d}{\eps} \right)^{d+1} \cdot n^{1/d} \log{n}  \right). 
$$
The proof of Theorem~\ref{T6} is now complete.
\end{proof}

\section{An asymptotically tight bound on the number of restricted
boxes}\label{sec:restricted} 

In this section we prove the following theorem:

\begin{theorem}\label{thm:restricted}
Let $U_d$ be the unit hypercube $[0,1]^d$.
For any $n > 0$, there exist $n$ points in $U_d$
such that the number of restricted boxes in $U_d$
is at least
$(\lfloor \frac{n}{d} \rfloor + 1)^d$.
On the other hand,
the number of restricted boxes determined by any set of $n$ points in $U_d$
is at most
${n \choose d} \cdot {2d \choose d}$.
\end{theorem}

We prove the lower bound in Theorem~\ref{thm:restricted} by construction.
We will use the following lemma:

\begin{lemma}\label{lem:construction}
Let $n = \sum_{i=1}^d n_i$, where $n_i \ge 2$, $1 \le i \le d$.
Then there exist $n$ points in $\RR^d$
such that the number of maximal empty axis-parallel boxes in $\RR^d$
is at least $\prod_{i=1}^d (n_i - 1)$.
\end{lemma}
\begin{proof}
Let $\pm\vec x_1, \ldots, \pm\vec x_d$
be the positive and negative unit vectors along the $d$ axes of $\RR^d$.
Partition these $2d$ vectors into $d$ groups of orthogonal vectors,
$$
\{+\vec x_1, -\vec x_2\}, \{+\vec x_2, -\vec x_3\},
	\ldots, \{+\vec x_{d-1}, -\vec x_d\}, \{+\vec x_d, -\vec x_1\},
$$
with one positive vector and one negative vector in each group.
Then, for each group of two orthogonal vectors, say $\{+\vec x_i, -\vec x_j\}$,
place a sequence of $n_i$ points in $\RR^d$ as
$$
k \vec x_i - (n_i + 1 - k) \vec x_j, \quad 1 \le k \le n_i,
$$
where each pair of consecutive points in the sequence,
say
$$(a - 1) \vec x_i - b \vec x_j
\quad\textup{and}\quad
a \vec x_i - (b - 1) \vec x_j,
$$
corresponds to a pair of open half-spaces
$$
x_i < a
\quad\textup{and}\quad
x_j > -b.
$$
\old{
The intersection of the two half-spaces
contains all points in the other sequences but no points in this sequence,
and is bounded by the two points
in the two directions $+\vec x_i$ and $-\vec x_j$.
There are $\prod_{i=1}^d (n_i - 1)$ combinations of
$d$ pairs of consecutive points,
one pair from each sequence.
For each combination,
the intersection of the corresponding $d$ pairs of half-spaces
is a unique maximal empty axis-parallel box.
}% old
 
Consider the pair of open half-spaces $x_i < a$ and $x_j > -b$
corresponding to the pair of consecutive points in the sequence for
the group $\{ +\vec x_i, -\vec x_j \}$.
Since the points in the sequence have monotonic $x_i$ and $x_j$ coordinates,
we have property (i) that
the intersection of the two half-spaces contains no points in the sequence,
and property (ii) that
each of the two points is on the boundary of one half-space
and is in the interior of the other half-space.
Moreover,
since the $x_i$ and $x_j$ coordinates of the points in the other sequences
are either zero or different in sign from the points in this sequence,
we have (iii) that
each of the two half-spaces contains all points in the other sequences.
There are $\prod_{i=1}^d (n_i - 1)$ combinations of $d$ pairs of consecutive
points, one pair from each sequence.
Consider the intersection $R_d$ of the $d$ pairs of half-spaces
corresponding to any of these combinations.
By (i), the intersection $R_d$ must be empty.
By (ii) and (iii),
there is a point in the interior of each bounding face,
thus the intersection box $R_d$ must be maximal.
Hence for each combination,
the intersection of the corresponding $d$ pairs of half-spaces
is a unique maximal empty axis-parallel box.
\begin{figure}[htb]
\psfrag{x}{$x$}
\psfrag{y}{$y$}
\centering\includegraphics[width=0.33\linewidth]{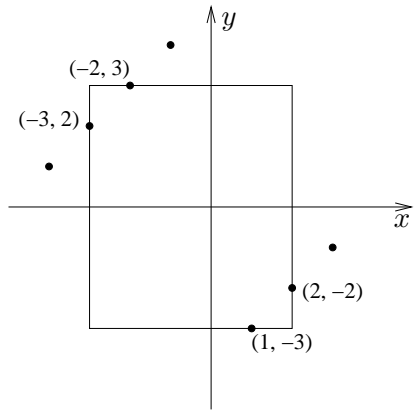}
\caption{An example of the construction.}
\label{fig:construction}
\end{figure}

We refer to Fig.~\ref{fig:construction} for an example of the planar case.
For $n_1 = 3$, $n_2 = 4$, and $n = 7$,
the four unit vectors $\pm\vec x$ and $\pm\vec y$ are grouped into
$\{+\vec x, -\vec y\}$ and $\{+\vec y, -\vec x\}$.
The corresponding two sequences of points have the following
$(x,y)$-coordinates:  
\begin{gather*}
(1, -3) \quad (2, -2) \quad (3, -1)
\\
(-4, 1) \quad (-3, 2) \quad (-2, 3) \quad (-1, 4).
\end{gather*}
Then the following two pairs of consecutive points
\begin{gather*}
(1, -3) \quad (2, -2)
\\
(-3, 2) \quad (-2, 3)
\end{gather*}
correspond to the following two pairs of half-planes:
\begin{gather*}
x < 2
\quad\textup{and}\quad
y > -3
\\
y < 3
\quad\textup{and}\quad
x > -3
\end{gather*}
whose intersection is the maximal empty box $(-3, 2) \times (-3, 3)$.
\end{proof}

By scaling and translation, the $n$ points in Lemma~\ref{lem:construction}
can be placed in the unit hypercube $U_d = [0,1]^d$ such that
the number of restricted boxes inside $U_d$ is at least
$\prod_{i=1}^d (\lfloor \frac{n}{d} \rfloor + 1)
= (\lfloor \frac{n}{d} \rfloor + 1)^d$,
where the change from $-1$ to $+1$ in the product
accounts for the two bounding faces of the unit hypercube
perpendicular to each axis.
This proves the lower bound. The same lower bound was obtained
independently and simultaneously by Backer and Keil~\cite{Ba09,BK09a,BK09b}.

To prove the upper bound in Theorem~\ref{thm:restricted},
we borrow the deflation-inflation idea of Backer and Keil~\cite{BK09a,BK09b}.
Assume for simplicity that
the points have distinct coordinates along each axis
(it is possible to perturb the points symbolically so this condition holds).
Let $B$ be an arbitrary restricted box.
Consider the $2d$ faces of the box in any fixed order.
If a face contains a point in its interior,
deflate the box by pushing the face toward its opposite face
until it contains a point on its boundary.
After $d$ such deflations,
we obtain an empty box $B'\subset B$ that is the smallest box containing
exactly $d$ points on its boundary.
To recover the original box $B$ from $B'$,
it suffices to inflate the box at the $d$ faces in reverse order,
by pushing each face away from its opposite face until
it contains a point in its interior.
Therefore the number of restricted boxes $B$ is at most
the number of deflated boxes $B'$ times the number of subsets of $d$
deflated faces, that is, ${n \choose d} \cdot {2d \choose d}$.
Since 
${n \choose d} \le n^d / d!$
and
${2d \choose d} = (2d)! / (d!)^2$,
we have
$$
{n \choose d} \cdot {2d \choose d}
\le n^d\frac{(2d)!}{(d!)^3}.
$$
By Stirling's formula, $d! = \sqrt{2\pi d} (d/e)^d (1 + O(1/d))$,
hence 
$$
\frac{(2d)!}{(d!)^3}
= \frac{\sqrt{2\pi 2d} (2d/e)^{2d}}{(\sqrt{2\pi d} (d/e)^d)^3}
	\big(1 \pm O(1/d)\big)
= \frac{(4e/d)^d}{\sqrt2 \pi d}
	\big(1 \pm O(1/d)\big).
$$
Thus
$$
{n \choose d} \cdot {2d \choose d}
\le n^d \frac{(4e/d)^d}{\sqrt2 \pi d}
	\big(1 \pm O(1/d)\big).
$$

Our upper bound is sharper (with respect to the dependence on $d$)
than the upper bound of $O(n^d) \cdot 2^{2d}$ by Backer and 
Keil~\cite{BK09a,BK09b}. The ratio of our upper bound to the lower bound is
$$
f(n,d)
= \frac{{n \choose d} \cdot {2d \choose d}}{(\lfloor \frac{n}{d} \rfloor + 1)^d}
\le \frac{n^d \frac{(4e/d)^d}{\sqrt2 \pi d} \big(1 \pm O(1/d)\big)}{(n/d)^d}
= \frac{(4e)^d}{\sqrt2 \pi d} \big(1 \pm O(1/d)\big)
= O\big((4e)^d/d\big).
$$
In comparison, the ratio of their upper bound to the same lower bound is 
$$
g(n,d)
= \frac{O(n^d) \cdot 2^{2d}}{(\lfloor \frac{n}{d} \rfloor + 1)^d}
= O\left(\frac{(4n)^d}{(n/d)^d}\right)
= O\big((4d)^d\big).
$$

\end{document}